\documentclass[12pt,reqno]{amsart}
\usepackage{amsfonts}

\usepackage{amscd,amssymb,amsthm}
\newcommand{\Tr}{\mathop{\rm Tr}\nolimits}
\newtheorem{theorem}{Theorem}[section]

\newtheorem{proposition}[theorem]{Proposition}
\newtheorem{corollary}[theorem]{Corollary}

\begin{document}

\author{M. CORGINI$^1$}
\author{C. ROJAS-MOLINA$^2$}
\author{D.P. SANKOVICH$^3$}
\title[COEXISTENCE OF NON-CONVENTIONAL CONDENSATES]{ COEXISTENCE OF NON-CONVENTIONAL CONDENSATES\\
IN TWO-LEVEL BOSE ATOM SYSTEM}
\address{{$^1$\it Departamento de Matem\'aticas, Universidad de La Serena,Cisternas 1200,
La Serena. Chile\\
mcorgini@yahoo.com}}
\address{{$^2$\it Departamento de Matem\'aticas, Universidad de La Serena,Cisternas 1200,
La Serena. Chile}}
\address{{$^3$\it Steklov Mathematical Institute,
Gubkin str. 8, 117966, Moscow, Russia\\
sankovch@mi.ras.ru}}

\maketitle

\begin{abstract}
In the framework of the Bogolyubov approximation and using the Bogolyubov
inequalities we give a simple proof of the coexistence of two
non-conventional Bose--Einstein condensates in the case of some
superstable Bose system whose atoms have an internal two-level
structure and their energy operators in the second quantized form
depend on the number operators only.
\end{abstract}

\section{Introduction}

\label{sec1}

The phenomenon of Bose--Einstein condensation (BEC) was described
first by Einstein \cite{eins}. Until recently the best
experi\-mental evidence that BEC could occur in a real physical
system was the phenomenon of superfluidity in liquid helium as
suggested originally by London who introduced the concept of
macroscopic occupation of the ground state~\cite{lon1} and
conjectured that the momentum-space condensation of bosons is
enhanced by spatial repulsion between the particles~\cite{lon2}.
However, nowadays, there exists a considerable amount of
experimental evidence for BEC~\cite{and,brad}.

Given the difficulty of the problem of proving the existence of
BEC from a mathematical point of view, it is
desirable to have idealized models in which  one can develop concrete
scenarios for BEC. In this sense, we are interested in studying  Bose
systems whose energy operators consider repulsive mean interactions
represented by diagonal operators in  the occupation numbers. It
frequently leads to thermodynamically stable systems which  can be
classically understood.

We are considering a class of systems for which various  Bose
condensates, in the sense of macroscopic occupation of the ground
state, coexist \cite {pul1,pul2}. Specifically we are interested in
the study of thermodynamic behavior of some systems for which the
Bose atoms have an internal two-level structure. We prove that the
existence of two ground state levels, one of them with negative
energy, leads to an enhancement of condensation.

This paper is divided as follows. In section ~\ref{sec2} we describe
the basic mathematical notions associated to this kind of systems
and we study
the stability of a system of atoms with internal structure. In section~\ref%
{sec3} making use of the so-called Bogolyubov approximation \cite
{bog,gin} we give an extremely simple proof of non-conventional BEC
in the sense of macroscopic occupation of the ground state
\cite{berg,bru2,bru3}. We prove that the existence of two ground
state levels, one of them with negative energy, leads to an
enhancement of condensation in the sense of macroscopic occupation
of both internal states \cite{pul1,pul2}. The proposed approximation
(Bogolyubov approximation) is obtained by substitution of the zero
mode creation and annihilation Bose operators by suitable chosen
c-numbers. In this proof a significative role play a well known
variational Bogolyubov inequalities and the fact that the involved
Hamiltonians are written in terms of the number operators. It
enables us to give a simplified demonstration of thermodynamic
equivalence of the limit grand canonical pressures corresponding to
the energy operator of the system and the respective approximating
Hamiltonian.

 Atoms BEC's in two different hyperfine states confined in a single
 trap with a time-varying Raman coupling between the two levels are
 studied in Ref.13. The  Hamiltonian associated to such system,
 being $\hat{a}, \hat{b}$  the annihilation Bose operators corresponding to
 the two different states and under the two mode approximation is given
 as:

\begin{equation}\hat{H} = \hat{H}_a + \hat{H}_b + \hat{H}_ {\mathrm{int}} +
\hat{H}_{\mathrm{las}},
\end{equation}
\begin{equation}
\hat{H}_a =  \omega_a \hat{a}^{\dag}\hat{a} +
\frac{U_a}{2}{\hat{a}}^{\dag}{\hat{a}}^{\dag}\hat{a}\hat{a},
\end{equation}
\begin{equation}
\hat{H}_b =  \omega_b \hat{b}^{\dag}\hat{b} +
\frac{U_b}{2}{\hat{b}}^{\dag}{\hat{b}}^{\dag}\hat{b}\hat{b},
\end{equation}
\begin{equation}
\hat{H}_{\mathrm{int}} =
\frac{U_{ab}}{2}{\hat{a}}^{\dag}\hat{a}{\hat{b}}^{\dag}\hat{b},
\end{equation}
\begin{equation}
\hat{H}_{\mathrm{las}} = \Omega(t)( e^{i\phi t}{\hat{a}}^{\dag}\hat{b}+
e^{-i\phi t}{\hat{b}}^{\dag}\hat{a}),
\end{equation}
where $ \hat{H}_a, \hat{H}_b $ describe systems with self
interactions. $\hat{H}_{ab}$ is associated to collisions between
both systems and $U_a,U_b, U_{ab}$ are constants. Finally
$\hat{H}_{\mathrm{las}}$ represents the operator associated to the Raman
coupling.

 We shall prove by using techniques of Bogolyubov approximation that an analogous
scenario for the coexistence of different condensates can be
obtained in the case of superstable systems of two-level Bose-atoms,
described by energy operators of the mean field type, undergoing
self interactions and collisions. Obviously the Raman coupling is
excluded from the respective Hamiltonians.

\section{The Model}\label{sec2}

The one-particle free Hamiltonian corresponds to the operator
$S^l=-\frac{\triangle}{2}$ defined on a dense subset of the Hilbert
space ${\mathcal{H}}^l = L^2 (\Lambda_l)$, being $\Lambda_l = \left[
-\frac{l}{2}, \frac{l}{2} \right]^d \subset {\mathbb R}^d $ a cubic
box of boundary $\partial \Lambda_l$ and volume $V_l= l^d $. In
other words, the particles are confined to bounded regions. We
assume periodic boundary conditions under which $S^l$ becomes a
self-adjoint operator.

In this section it is assumed that the Bose atoms have an internal
two-level structure analogous to the SU(2) spin symmetry
\cite{pul1,pul2}. In this case any one-particle wave function has
the form $\phi\otimes s$ where, $\phi\in L^2 (\Lambda_l)$ and $s \in
{\mathbb C}^2$ represents the internal state. Therefore the vector
space associated to this system is in fact, ${\mathcal H}^l_s =
L^2(\Lambda_l)\otimes {\mathbb C}^2. $

We shall study the model of Bose particles whose Hamiltonian
is given as:
\begin{equation}  \label{eq1}
\hat{H}_l = \hat{H}^0_l + \hat{U}^I_{l} = \hat{H}^0_l +\frac{\gamma_0}{V_l}%
\displaystyle\sum_{\sigma,\mathbf{p}\in\Lambda^*_l} (\hat{a}^{\dag}_{\mathbf{p},\sigma})^2%
\hat{a}^2_{\mathbf{p},\sigma} +
\frac{\gamma}{V_l}\hat{n}_{\mathbf{0},-}\hat{n}_{\mathbf{0},+} + V_l
g\left(\frac{{\hat{N}}^{^{\prime}}}{V_l}\right),
\end{equation}
where $g(x)\geq ax^2, a>0$ is a continuous function and $\sigma = +$
or $ -$ depending on the corresponding label of internal energy. The
second term at the right hand side of eq.~(\ref{eq1}) represents the
intrastate collisions, or the self-scattering, the third term
represents the interstate collisions or the cross scattering.
The sum in (\ref{eq1}) runs over the set $%
\Lambda^*_l = \{ \mathbf{p}\in {\mathbb R}^d: p_{\alpha} =
2\pi n_{\alpha}/l,
n_{\alpha} = 0,\pm 1,\pm 2,,...\alpha =1,2,..d \} $, $\hat{a}%
^{\dag}_{\mathbf{p},\sigma}, \hat{a}_{\mathbf{p},\sigma}$ are the
Bose operators of creation and annihilation of particles defined on
the Fock space ${\mathcal{F}}_B$
and satisfying the usual commutation rules: $\lbrack \hat{a}%
_{\mathbf{q},\sigma_1},\hat{a}^{\dag}_{\mathbf{p},\sigma_2} \rbrack
= \hat{a}_{\mathbf{q},\sigma_1}
\hat{a}^{\dag}_{\mathbf{p},\sigma_2} - \hat{a}^{\dag}_{\mathbf{p},\sigma_2} \hat{a}%
_{\mathbf{q},\sigma_1} = \delta_{\mathbf{p},\mathbf{q}}\delta_{\sigma_1,\sigma_2}. $ $\hat{n}%
_{\mathbf{p},\sigma} = \hat{a}^{\dag}_{\mathbf{p},\sigma}
\hat{a}_{\mathbf{p},\sigma}$ is the number operator associated to
index $\mathbf{p}$ and internal label $\sigma$ and
\begin{equation}
\lambda_{l,\sigma}(\mathbf{p})= \left\{\begin{array}{ll}
\sigma \lambda ,& \mbox{ $ \mathbf{p}=\mathbf{0} $}, \\
\noalign{\smallskip} \frac{\left\|\mathbf{p}\right\|^2}{2},&\mbox{
$\mathbf{p}\neq \mathbf{0}.$}
\end{array}\right.
\end{equation}
In this case $\hat{H}^0_l = \sum_{\sigma, \mathbf{p}\in
{\Lambda }^*_l}
\lambda_{l,\sigma}(\mathbf{p})\hat{n}_{\mathbf{p},\sigma} $ and $\gamma_0 ,\gamma > 0 $. $\hat{%
N} = \sum_{\sigma,\mathbf{p}\in {\Lambda }^*_l} \hat{a}%
^{\dag}_{\mathbf{p},\sigma} \hat{a}_{\mathbf{p},\sigma}$ is the total number operator and $%
\hat{N}^{^{\prime}} = \sum_{\sigma,\mathbf{p}\in {\Lambda }%
^*_l\backslash \{ \mathbf{0} \}
}\hat{a}^{\dag}_{\mathbf{p},\sigma}\hat{a}_{\mathbf{p},\sigma}$ is
the total number operator with exclusion of
$\hat{n}_{\mathbf{0},\sigma}$.

Note that the boson Fock space ${\mathcal{F}}_B$ is isomorphic to the tensor
product $\displaystyle\otimes_{\sigma, \mathbf{p}\in \Lambda^*_l} {\mathcal{F}}%
^B_{\mathbf{p},\sigma} $ where ${\mathcal{F}}^B_{\mathbf{p},\sigma}
$ is the boson Fock space constructed on the one-dimensional Hilbert
space ${\mathcal{H}}_{\mathbf{p},\sigma} = \{ \theta
e^{i\mathbf{p}\cdot\mathbf{x}}\otimes e_{\sigma} \}_{\theta\in
{\mathbb C}},$ where $e_{-} = (0,1)$ and $e_{+} = (1,0).$

Let
\begin{equation}  \label{eq2}
p_l(\beta, \mu) = \frac{1}{\beta V_l}\ln \mathop{\rm Tr}\nolimits_{\mathcal{F%
}_B} \exp \left[ -\beta( \hat{H}_l - \mu \hat{N})\right]
\end{equation}
be the grand-canonical pressure corresponding to $\hat{H}_l$, where $\beta =
\theta^{-1}$ is the inverse temperature.
If $\hat{H}_l(\mu)= \hat{H}_l - \mu \hat{N}$, the equilibrium Gibbs state
(grand canonical ensemble)$\langle - \rangle_{\hat{H}_l(\mu)}$ is defined as

\begin{equation}  \label{eq3}
\langle \hat{A} \rangle_{\hat{H}_l(\mu)} = \left[ \mathop{\rm Tr}\nolimits_{%
\mathcal{F}_B} \exp \left( -\beta \hat{H}_l (\mu)\right)\right]^{-1}%
\mathop{\rm Tr}\nolimits_{\mathcal{F}_B} \left[\hat{A} \exp \left( -\beta \hat{H}%
_l (\mu)\right)\right],
\end{equation}
for any operator $\hat{A}$ acting on the symmetric Bose--Fock space.
The total density of particles $\rho(\mu)$ for infinite volume is
defined as
\begin{equation}  \label{eq4}
\displaystyle\lim_{V_l\to\infty}\left\langle \frac{\hat{N}}{V_l}\right\rangle_{\hat{H}%
_l(\mu)} = \displaystyle\lim_{V_l\to\infty}\rho_{l}(\mu) = \rho (\mu),
\end{equation}
and the density of particles $\rho_{\mathbf{p},\sigma}(\beta, \mu)$
associated to the energy label $\lambda_{\sigma} (\mathbf{p}) $ is
defined as
\begin{equation}  \label{eq5}
\displaystyle\lim_{V_l\to\infty}\left\langle
\frac{\hat{n}_{\mathbf{p},\sigma}}{V_l}
\right\rangle_{\hat{H}_l(\mu)} = \displaystyle\lim_{V_l\to\infty}\rho_{\mathbf{p},%
\sigma,l}(\beta,\mu) = \rho_{\mathbf{p},\sigma}(\beta, \mu).
\end{equation}

We shall say that the system undergoes a macroscopic occupation of
the single $(\mathbf{p},\sigma)$-mode particle level if
$\rho_{\mathbf{p},\sigma}(\beta, \mu) > 0.$

In the case of models with a two label internal structure,
\cite{pul1,pul2} two macroscopic occupations independent on
temperature (non-con\-ven\-ti\-onal BEC) in the sense that
\begin{equation}
\rho_{\mathbf{0},-}(\mu) > 0,\quad\rho_{\mathbf{0},+}(\mu) > 0
\end{equation}
can coexist.

We are interested in proving the superstability of the
model, in other words we verify that the limit  pressure $%
\displaystyle\lim_{V_l\to\infty}p_{l}(\beta,\mu)$ exists for all
$\mu\in {\mathbb R}.$ A useful criterion ~\cite{rue} says that we
are in presence of a superstable system if given an interaction
$\hat{U}^{\mathrm I}_l$ defined on the Bose Fock space ${\mathcal F
}_B,$ there exist two constants $C_1
> 0$ and $C_2 \geq 0$ such that
\begin{equation}
\hat{U}^{\mathrm I}_l \geq -\frac{C_2}{V_l}\hat{N} + \frac{C_1}{V_l}\hat{N}^2.
\end{equation}

Then, the superstability of the system whose energy operator has the form ~(%
\ref{eq1}) is ensured by the following result:

\begin{proposition}
\label{thm1} The operator $\hat{U}^{\mathrm I}_l $ given by ~(\ref{eq1}) satisfies the
inequality
\begin{equation}  \label{eq6}
\hat{U}^{\mathrm I}_l \geq -\frac{\gamma_0}{V_l}\hat{N} + \frac{\min\{a,\gamma_0\}}{%
4V_l}\hat{N}^2.
\end{equation}
\end{proposition}

\begin{proof}
For $x,y\in {\mathbb R}$ we have $x^{2} + y^{2} \geq
(x+y)^{2}/2.$ Moreover, since
$(\hat{a}^{\dag}_{\mathbf{p},\sigma})^{2}(\hat{a}_{\mathbf{p},\sigma})^{2}
= \hat{n}_{\mathbf{p},\sigma} (\hat{n}_{\mathbf{p},\sigma}-1 ) \geq
0$ and the spectrum  ${\mathrm {Sp}}(\hat{n}_{\mathbf{p},\sigma})=
{\mathbb N}\cup \{0\},$ for all $\sigma = \pm, \mathbf{p}\in
\Lambda^*_l,$ we obtain, in the sense of operators, the following
sequence of inequalities:

\begin{eqnarray}
\hat{U}^{\mathrm I}_l&\geq -\dfrac{\gamma_0}{V_l} \hat{n}_{\mathbf{0},-} -\dfrac{\gamma_0}{V_l} \hat{n}_{\mathbf{0},+}+\dfrac{\gamma_0}{V_l}\hat{n}^2_{\mathbf{0},-} + \dfrac{\gamma_0}{V_l}\hat{n}
^2_{\mathbf{0},+}+ V_l g(\dfrac{{\hat{N}}^{^{\prime}}}{V_l})\nonumber \\
&\geq -\dfrac{\gamma_0}{V_l} \hat{N}+ \dfrac{\gamma_0}{V_l}\hat{n}^2_{\mathbf{0},-} +
\dfrac{\gamma_0}{V_l}\hat{n}^2_{\mathbf{0},+}+a\dfrac{(\hat{N}^{^{\prime}})^2}{V_l}\nonumber \\
&\geq -\dfrac{\gamma_0}{V_l}
\hat{N}+\dfrac{\gamma_0}{V_l}\hat{n}^2_{\mathbf{0},-}+
\dfrac{\min\{a,\gamma_0\}}{V_l}[\hat{n}^2_{\mathbf{0},+} +(\hat{N}^{^{\prime}})^2] \nonumber\\
&\geq -\dfrac{\gamma_0}{V_l}
\hat{N}+\dfrac{\gamma_0}{2V_l}\hat{n}^2_{\mathbf{0},-}+
\dfrac{\min\{a,\gamma_0\}}{2V_l}[\hat{n}_{\mathbf{0},+} +\hat{N}^{^{\prime}}]^2 \nonumber\\
&\geq -\dfrac{\gamma_0}{V_l} \hat{N}+\dfrac{\min\{a,\gamma_0\}}{2V_l}[\hat{n}
^2_{\mathbf{0},-}+[\hat{n}_{\mathbf{0},+} +\hat{N}^{^{\prime}}]^2] \nonumber\\
&\geq -\dfrac{\gamma_0}{V_l} \hat{N}+\dfrac{\min\{a,\gamma_0\}}{4V_l}[\hat{n}
_{\mathbf{0},-}+ \hat{n}_{\mathbf{0},+} +\hat{N}^{^{\prime}}]^2 \nonumber\\
&=-\dfrac{\gamma_0}{V_l}
\hat{N}+\dfrac{\min\{a,\gamma_0\}}{4V_l}\hat{N}^2.
\end{eqnarray}
\noindent
This concludes the proof.
\end{proof}

\section{Approximating Hamiltonian}\label{sec3}

For any complex number $c\in \mathbb{C} $, let consider the coherent vector $%
\phi_{c,\sigma} $ in $\mathcal{F}_{0,\sigma},$ given as
\begin{equation}  \label{eq8}
\phi_{c,\sigma} = e^{-V_l\left| c\right|^{2}/2} \displaystyle%
\sum^{\infty}_{n=0}\frac{1}{n!}\left( V^{\frac{1}{2}}_l c\right)^n
\left( \hat{a}^{\dag}_{\mathbf{0},\sigma}\right)^{n} \phi
_{\mathbf{0},\sigma} ,
\end{equation}
where $\phi_{\mathbf{0},\sigma} = \phi_{\mathbf{0}}\otimes
e_{\sigma} $ is the vacuum vector
in $\mathcal{F}_{\mathbf{0},\sigma}.$ Then, for any operator $\hat{A}$ defined in $%
\mathcal{F}_{B}$, it is possible to construct an operator
$\hat{A}\left( c\right) $ such that
\begin{equation}  \label{eq9}
\left\langle \Psi _{1}^{\prime }, \hat{A}\left( c\right) \Psi
_{2}^{\prime }\right\rangle_{\mathcal{F}_B} =\left\langle \Psi
_{1}^{\prime }\otimes
\phi_{c,\sigma}, \hat{A} \Psi _{2}^{\prime }\otimes \phi_{c,\sigma} \right\rangle_{%
\mathcal{F}_B},
\end{equation}
where $\Psi _{1}^{\prime },\Psi _{2}^{\prime }\in \mathcal{F}^{^{\prime}}_B$%
, being $\mathcal{F}^{^{\prime}}_B = \displaystyle\otimes_{\sigma,
\mathbf{p}\in \Lambda^*\backslash \{\mathbf{0}\}}
{\mathcal{F}}^B_{\mathbf{p},\sigma}.$ The transition from operator
$\hat{A}$ to operator $\hat{A}\left( c\right) $ is known as
Bogolyubov approximation and it consists in replacing the operators $\hat{a}%
^{\dag}_{\mathbf{0}} $ and $\hat{a}_{\mathbf{0}} $ in any $\hat{A}$
expressed in the normal form ( all creation operators are grouped on
the left hand side of the group of
annihilation operators in $\hat{A}$) operators with the complex numbers $%
\sqrt{V}\bar{c}$ and $\sqrt{V}c $. Then, in the framework of the Bogolyubov approximation,
replacing the operators $ \hat{a}%
^{\dag}_{\mathbf{0},\sigma} $ and $\hat{a}_{\mathbf{0},\sigma} $ in
$\hat{H}_l (\mu)$ with the
complex numbers $\sqrt{V_l}\bar{c},$ $\sqrt{V_l}c, $ for $\sigma = -$ and $%
\sqrt{V_l}\bar{\eta},$ $\sqrt{V_l}\eta $, for $\sigma = +$, respectively, we
get the following approximating Hamiltonian $\hat{H}^{\mathrm{appr}}_l (c,\eta,\mu) $,

\begin{eqnarray}
\hat{H}^{\mathrm{appr}}_l (c,\mu) &=
\displaystyle\sum_{\sigma,\mathbf{p}\in\Lambda^*_l\backslash
\{\mathbf{0}\}} (\lambda_{l,\sigma}(\mathbf{p})-\mu
)\hat{n}_{\mathbf{p},\sigma} +
\frac{\gamma_0}{V_l}\displaystyle\sum_{\sigma,\mathbf{p}\in\Lambda^*_l\backslash
\{\mathbf{0}\}}
(\hat{a}^{\dag}_{\mathbf{p},\sigma})^2\hat{a}^2_{\mathbf{p},\sigma}\nonumber\\
&+
V_l g\left(\frac{{\hat{N}}^{^{\prime}}}{V_l}\right) +(\lambda_{l,-}
(\mathbf{0})-\mu ) V_l \left|c\right|^2 + \gamma_0 V_l
\left|c\right|^{4}\nonumber\\
&+ (\lambda_{l,+}(\mathbf{0})-\mu ) V_l
\left|\eta\right|^2 + \gamma_0 V_l \left|\eta\right|^{4}+ \gamma V_l
\left|c\right|^2\left|\eta\right|^2,
\end{eqnarray}%

with $\;c,\eta\in {\mathbb C}.$

Note that $\hat{H}^{\mathrm{appr}}_l (c,\mu) $ can be rewritten in a more
suggestive form as:
\begin{eqnarray} \label{eq10}
\hat{H}^{\mathrm{appr}}_l (c,\mu) = & \hat{H}^{\mathrm{MF}^{'}}(\mu)
+\dfrac{\gamma_0}{V_l}\displaystyle\sum_{\sigma,\mathbf{p}\in\Lambda^*_l\backslash
\{\mathbf{0}\}}
(\hat{a}^{\dag}_{\mathbf{p},\sigma})^2\hat{a}^2_{\mathbf{p},\sigma} +(\lambda_{l,-} (\mathbf{0})-\mu ) V_l \left|c\right|^2 + \gamma_0 V_l
\left|c\right|^{4}\nonumber\\&+ (\lambda_{l,+} (\mathbf{0})-\mu ) V_l
\left|\eta\right|^2 + \gamma_0 V_l \left|\eta\right|^{4}+ \gamma V_l
\left|c\right|^2\left|\eta\right|^2,
\end{eqnarray}
where
\begin{equation}
\hat{H}^{\mathrm{MF}^{'}}(\mu)\equiv
\sum_{\sigma,\mathbf{p}\in\Lambda^*_l\backslash
\{\mathbf{0}\}} (\lambda_{l,\sigma}(\mathbf{p})-\mu
)\hat{n}_{\mathbf{p},\sigma}+
V_l g\left(\frac{{\hat{N}}^{^{\prime}}}{V_l}\right)
\end{equation}
denotes the so-called mean field Hamiltonian with excluded zero mode.

In this case we have the following theorem \cite{gin}.
\begin{theorem}\label{thm2}
The following inequality takes place:
\begin{equation}  \label{eq11}
\mathop{\rm \Tr}\nolimits_{ \mathcal{F}^{^{\prime}}_B } e^{-\beta \hat{H}%
^{\mathrm{appr}}_l (c,\eta,\mu) }\leq \mathop{\rm Tr}\nolimits_{ \mathcal{F}_B }
e^{-\beta \hat{H}_l (\mu) }.
\end{equation}
\end{theorem}

\begin{proof}
From the above definitions it follows that
\begin{eqnarray}
& \mathop{\rm\Tr}\nolimits_{\mathcal{F}^{\prime}_{B }}\left(e^{-\beta \hat{%
H}_l (c,\eta,\mu)}\right) = \displaystyle\sup_{\left\{\Psi^{\prime}_{n}\right\}}
\displaystyle\sum_{n}\exp \left( -\beta \left< \Psi^{\prime}_{n}, \hat{H}_{l}\left( c,\eta,\mu\right) \Psi^{\prime}_{n}
\right> \right) \nonumber\\
& = {\displaystyle\sup_{\left\{ \Psi _{n}^{\prime }\right\}} }\displaystyle
\sum_{n}\exp \left( -\beta \left\langle \Psi_{n}^{\prime }\otimes
\phi_{c,-}\otimes \phi_{\eta,-}, \hat{H}_{l}(\mu) \Psi_{n}^{\prime
}\otimes
\phi_{c,-}\otimes \phi_{\eta,+}\right\rangle \right) \nonumber\\
&\leq \displaystyle\sup_{\left\{ \Psi _{n}^{\prime }\right\}} \displaystyle
\sum_{n}\left\langle \Psi _{n}^{\prime }\otimes \phi_{c,-}\otimes
\phi_{\eta,+}, e^{-\beta \hat{H}_{l}(\mu)}\Psi_{n}^{\prime }\otimes
\phi_{c,-}\otimes \phi_{\eta,+}\right\rangle \nonumber\\
&\leq \mathop{\rm \Tr}\nolimits_{\mathcal{F}_{B}}\left( e^{-\beta
\hat{H}_{l }(\mu)}\right).
\end{eqnarray}
The supremum is over all orthonormal basis  of the corresponding
Fock space $\mathcal{F}^{\prime}_B.$ The first inequality is the
Peierls inequality; the second one is obtained taking the coherent
vector $\phi_{c,-}\otimes\phi_{\eta,+} $ as one of the vectors in
the orthonormal basis of $\mathcal{F}_{\mathbf{0}}
=\mathcal{F}_{\mathbf{0},-}\otimes \mathcal{F}_{\mathbf{0},+}.$ This
complets the proof.
\end{proof}

\subsection{Bogolyubov Inequalities}

Let $\hat{H}_{a,l}$ and $\hat{H}_{b,l}$ be self-adjoined operators defined on $%
\mathcal{D} \subset {\mathcal{F}}_B$.  $p_{a,l}(\beta,\mu), $ $%
p_{b,l}(\beta,\mu) $ represent the grand canonical pressures
corresponding to the operators  $\hat{H}_{a,l}$, $\hat{H}%
_{b,l}$. In this case the following well known Bogolyubov inequalities
\begin{eqnarray} \label{eq12}
\left\langle\dfrac{\hat{H}_{a,l}(\mu)-\hat{H}_{b,l}(\mu)}{V_l}\right\rangle_{\hat{H}%
_{a,l}(\mu)} &\leq p_{b,l} (\beta,\mu) - p_{a,l} (\beta,\mu)\nonumber\\
& \leq \left\langle%
\dfrac{\hat{H}_{a,l}(\mu) - \hat{H}_{b,l}(\mu)}{V_l}\right\rangle_{\hat{H}%
_{b,l}(\mu) },
\end{eqnarray}
hold, where $\langle - \rangle_{\hat{H}_{a,l}(\mu)},$ $\langle - \rangle_{%
\hat{H}_{b,l}(\mu)}$  are the Gibbs states in the grand canonical ensemble
associated to the Hamiltonians $\hat{H}_{a,l}, \hat{H}_{b,l}, $ respectively.

\subsection{Global and Local Minimums for Convex
Functions}
\begin{proposition}\label{thm3} Let $G:U\rightarrow {\mathbb R}$ be
convex on a convex set $ U\subset L$ ( $L$ is a normed linear
space). If $G$ has a local minimum at $\bar{x},$ then $G(\bar{x})$
is also a global minimizer. The set $V$ ( conceivably empty ) on
which $G$ attains its minimum is convex. And if $G$ is strictly
convex in a neighbourhood of a minimum point $\hat{x},$ then
$V=\{\hat{x}\};$ that is , the minimum  is unique.
\end{proposition}
The proof is contained in Ref.15.

\subsection{Equivalence of Limit Grand Canonical Pressures. }

\begin{theorem}
\label{thm4} Under any of the conditions

 (i) $2\gamma_0 >\gamma, $ $\mu \geq\lambda(2\gamma_0 +\gamma)/(2\gamma_0 - \gamma),$

 (ii) $\gamma = 0,$ $\mu \in {\mathbb R},$

 (iii) $ 2\gamma_0 = \gamma, $ $\mu \in {\mathbb R}, $

\noindent the systems with Hamiltonians $\hat{H}_l(\mu)$ and
$\hat{H}^{\mathrm{appr}}_l (c,\eta,\mu)$ are thermodynamically
equivalent in the following sense,
\begin{equation} \label{eq13}
\displaystyle\lim_{V_l\to\infty} \sup_{\left|c\right|, \left| \eta
\right|: c,\eta \in {\mathbb C}} p^{\mathrm{appr}}_l (\beta,c,\eta, \mu) =
p(\beta,\mu),
\end{equation}
where $p^{\mathrm{appr}}_l (\beta,c,\eta, \mu) = (\beta V_l)^{-1}\ln \mathop{\rm Tr}%
\nolimits_{\mathcal{F}_B} \exp \left[-\beta( \hat{H}^{\mathrm{appr}}_l(c,\eta,
\mu))\right]$.
\end{theorem}

\begin{proof}

Applying the Bogolyubov inequalities for pressures ~(\ref{eq12}) and inequality ~(%
\ref{eq11}), we obtain that:

\begin{equation} \label{eq14}
0\leq p_l(\beta,\mu) - p^{\mathrm{appr}}_l (\beta,c,\eta, \mu) \leq \frac{1}{V_l}
\langle \hat{H}^{\mathrm{appr}}_l (c,\eta, \mu) - \hat{H}_l(\mu) \rangle_{\hat{H}
_l(\mu)}\equiv \Delta_l(c,\eta,\mu).
\end{equation}
The right hand side of inequality ~(\ref{eq14}) can be rewritten as,

\begin{eqnarray}
\Delta_l(c,\eta,\mu) &= (\lambda_{l,-} (\mathbf{0})-\mu )
\left|c\right|^2 + \gamma_0\left|c\right|^{4}+ (\lambda_{l,+}
(\mathbf{0})-\mu )\left|\eta\right|^2 +
\gamma_0\left|\eta\right|^{4} + \gamma
\left|c\right|^2\left|\eta\right|^2 \nonumber\\&+ \left\langle
\left(\displaystyle\sum_{\sigma=\pm}(\mu-\lambda_{l,\sigma}
(\mathbf{0}))\hat{\rho}_{\mathbf{0},\sigma} -
\gamma_0\hat{\rho}_{\mathbf{0},\sigma}
(\hat{\rho}_{\mathbf{0},\sigma} - \frac{1}{V_l})\right) -\vphantom{\displaystyle\sum_{\sigma=\pm}(\mu-\lambda_{l,\sigma}
(\mathbf{0}))\hat{\rho}_{\mathbf{0},\sigma}} \gamma \hat{\rho}_{\mathbf{0},-} \hat{\rho}%
_{\mathbf{0},+}\right\rangle_{\hat{H}_l(\mu)}
\end{eqnarray}
where $\hat{\rho}_{\mathbf{0},\sigma} =
\hat{n}_{\mathbf{0},\sigma}/V_l.$

(i) Under the condition $2\gamma_0 >\gamma $ and according
proposition~\ref{thm3} $g(x,y):{\mathbb R}^2\rightarrow {\mathbb R}$
given as:
\begin{equation}
g(x,y) = (\lambda_{l,-} (\mathbf{0})-\mu )x +\gamma_0 x^2 +
(\lambda_{l,+} (\mathbf{0})-\mu )y +\gamma_0 y^2 + \gamma xy,
\end{equation}
is a  strictly convex function attaining a global minimum at
\begin{equation}
x_{\mathbf{0}} = \frac{\gamma (\lambda_{l,+} (\mathbf{0})-\mu) - 2\gamma_0 (\lambda_{l,-} (\mathbf{0})-\mu )}{%
4\gamma^2_0 -\gamma^2},
\end{equation}
\begin{equation}
y_{\mathbf{0}} = \frac{\gamma (\lambda_{l,-} (\mathbf{0})-\mu) - 2\gamma_0 (\lambda_{l,+} (\mathbf{0})-\mu )}{
4\gamma^2_0 -\gamma^2},
\end{equation}
with $(x,y) \in {\mathbb R}^2,$ since $\partial^2_{x} g = 2\gamma_0
> 0$ and $\partial_x g \partial_y g-[
\partial^2_{xy} g]^2 = 4\gamma^2_0 - \gamma^2.$
Taking $\lambda_{l,-} (\mathbf{0})= -\lambda$ and $\lambda_{l,+}
(\mathbf{0}) = \lambda,$ we have that $g(x,y) = -( \lambda+\mu )x
+\gamma_0 x^2 + (\lambda -\mu )y +\gamma_0 y^2 + \gamma xy$ and
\begin{equation}  \label{eq15}
x_{\mathbf{0}} = |c_{\mathbf{0}}|^2 = \frac{\lambda}{
2\gamma_0-\gamma} + \frac{\mu}{2\gamma_0 + \gamma},
\end{equation}
\begin{equation}\label{eq16}
y_{\mathbf{0}} = |\eta_{\mathbf{0}}|^2 = -\frac{\lambda}{
2\gamma_0-\gamma} + \frac{\mu}{2\gamma_0 + \gamma},
\end{equation}
for $\mu \geq \lambda (2\gamma_0 +\gamma)/(2\gamma_0 - \gamma).$
It is easy to verify that:
\begin{equation} g(x_{\mathbf{0}},y_{\mathbf{0}}) = -\left(\frac{\lambda^2}{2\gamma_0-\gamma}
+\frac{\mu^2}{2\gamma_0+\gamma}\right).
\end{equation}
On the other hand, being $\mu_l = \mu +\gamma_0/V_l$ we have
\begin{eqnarray}
\langle g^*_l(\hat{\rho}_{\mathbf{0},-},\hat{\rho}_{\mathbf{0},+})
\rangle_{\hat{H}_l(\mu)} & = \left\langle \left(
\displaystyle\sum_{\sigma =\pm}(\lambda_{l,\sigma}
(\mathbf{0})-\mu_l )\hat{\rho}_{\mathbf{0},\sigma} +\gamma_0
\hat{\rho} ^2_{\mathbf{0},\sigma}\right)\right. +\left.\vphantom{\displaystyle\sum_{\sigma =\pm}(\lambda_{l,\sigma}} \gamma
\hat{\rho}_{\mathbf{0},-}\hat{\rho}_{\mathbf{0},+}
\right\rangle_{\hat{H}_l(\mu)}
\end{eqnarray}
where
\begin{equation}
g^*_l(x,y) = (\lambda_{l,-} (\mathbf{0})-\mu_l )x +\gamma_0 x^2 +
(\lambda_{l,+} (\mathbf{0})-\mu_l)y +\gamma_0 y^2 + \gamma xy.
\end{equation}
With $\lambda_{l,-} (\mathbf{0})= -\lambda, \lambda_{l,+}
(\mathbf{0}) = \lambda,$ and under the same before established
conditions $g^*_l$ has a global minimum at the positive real values
\begin{equation}
x^*_{\mathbf{0},l} = \frac{\lambda}{ 2\gamma_0-\gamma} +
\frac{\mu_l}{2\gamma_0 + \gamma},
\end{equation}
\begin{equation}
y^*_{\mathbf{0},l} = -\frac{\lambda}{ 2\gamma_0-\gamma} +
\frac{\mu_l}{2\gamma_0 + \gamma}.
\end{equation}
Clearly $x^*_{\mathbf{0},l}\to x_{\mathbf{0}}$ and $y^*_{\mathbf{0},l}\to y_{\mathbf{0}} $ when $V_l\to \infty$ and $%
\displaystyle\lim_{V_l\to\infty}g^*_l(x^*_{\mathbf{0},l},y^*_{\mathbf{0},l})=
g(x_{\mathbf{0}},y_{\mathbf{0}}).$ Moreover, $g^*_l(x,y)\leq g(x,y)$
for all $(x,y)\in {{\mathbb R}^+ }^2,$ then
\begin{equation}
\displaystyle\inf_{(x,y)\in {{\mathbb R}^+ }^2} g^*_l(x,y)\leq \displaystyle%
\inf_{(x,y)\in {{\mathbb R}^+ }^2} g(x,y)
\end{equation}
which implies that
\begin{equation}
\displaystyle\inf_{(x,y)\in {{\mathbb R}^+  }^2} g(x,y) + \displaystyle\sup_{(x,y)\in {%
{\mathbb R }^+ }^2}\{- g^*_l(x,y)\}\geq 0.
\end{equation}
With these definitions we can rewrite $\Delta_l(c,\eta,\mu)$ as,
\begin{equation}
\Delta_l(c,\eta,\mu) = g (\left|\eta\right|^2,\left|c\right|^2) +
\langle
-g^*_l(\hat{\rho}_{\mathbf{0},-},\hat{\rho}_{\mathbf{0},+})\rangle_{\hat{H}_l(\mu)}.
\end{equation}
The energy operator ~(\ref{eq1}) is diagonal respect to the number
operators. Therefore, since the spectra of these operators coincide
with the set of non negative integers, this model can be classically
understood by using non negative random variables defined on a
suitable probability space $\Omega_l$. Let $\Omega_l$ be the
countable  set of sequences $\omega = \{\omega
(\mathbf{p},\sigma)\in {\mathbb N}: \sigma=\pm ,
 \mathbf{p}\in
\Lambda^*_l\}\subset
 {\mathbb N} \cup \{0\}$ satisfying
\begin{equation}
\sum_{\sigma,\mathbf{p}\in \Lambda^*_l}\omega
(\mathbf{p},\sigma)<\infty \;.
\end{equation}%
The basic random variables are the occupation numbers
$\{n_{\mathbf{p},\sigma}:j=1,2,....,\sigma=\pm \}$. They are defined
as the functions $ n_{\mathbf{p},\sigma}:\Omega _{l}\rightarrow
{\mathbb{N}}$ given as $n_{\mathbf{p},\sigma}(\omega )=\omega
(\mathbf{p},\sigma)$ for any $\omega \in \Omega_{l}$. The total
number of particles in the configuration $\omega $ is denoted as
$N(\omega ).$ Then the total number, excluded the zero mode is
denoted as $N^{\prime}(\omega ).$

Let $H_l(\mu)$ be the function of the random variables $
n_{\mathbf{p},\sigma} $ defined as
\begin{equation}H_l(\mu) =\left(
\displaystyle\sum_{\sigma,\mathbf{p}\in\Lambda^*_l}
(\lambda_{l,\sigma}(\mathbf{p})-\mu_l)n_{\mathbf{p},\sigma}
+\frac{\gamma_0}{V_l}n^2_{\mathbf{p},\sigma}\right) +
\frac{\gamma}{V_l}n_{\mathbf{0},-}n_{\mathbf{0},+} + V_l
g(\frac{N^{^{\prime}}}{V_l}).
\end{equation}
Let ${\mathbb P}$ be a probability density defined for any
$\omega\in\Omega_l$ as
\begin{equation}{\mathbb P}[\omega] =  \left[
\displaystyle\sum_{\omega\in\Omega}\exp \left( -\beta [H_l
(\mu)](\omega)\right)\right]^{-1}\exp \left( -\beta [H_l
(\mu)](\omega)\right).\end{equation}

Being $\mathbb E$ the expectation related to ${\mathbb P}$, the
concavity of $-g^*_l$ implies that this function has a global
maximum, then we obtain,

\begin{equation}
\langle-
g^*_l(\hat{\rho}_{\mathbf{0},-},\hat{\rho}_{\mathbf{0},+})\rangle_{\hat{H}_l(\mu)}
= {\mathbb E} (-g^*_l(\rho_{\mathbf{0},-},\rho_{\mathbf{0},+})) \leq
{\mathbb E} (\sup_{(x,y) \in {{\mathbb R}^+ }^2}\{-g^*_l(x,y)\}),
\end{equation}
where $\rho_{\mathbf{0},\sigma} =n_{\mathbf{0},\sigma}/V_l.$
From this it follows that
\begin{eqnarray}
& 0 \leq p_l(\beta,\mu)-\displaystyle\sup_{\left|c\right|, \left| \eta \right|:
c,\eta \in {\mathbb C}} p^{\mathrm{appr}}_l (\beta,c,\eta, \mu) \leq
\displaystyle\inf_{\left|c\right|, \left| \eta
\right|: c,\eta \in {\mathbb C}}\Delta_l(c,\eta,\mu)\nonumber \\
& \leq \displaystyle\inf_{\left|c\right|, \left| \eta \right|: c,\eta \in {\mathbb C }%
}g(\left|c\right|^2,\left|\eta\right|^2) + \displaystyle\sup_{(x,y) \in {{\mathbb
R}^+ }^2} \{-g^*_l(x,y)\}.
\end{eqnarray}%
Therefore,
\begin{equation}  \label{eq19}
0\leq \displaystyle\lim_{V_l\to\infty} \inf_{\left|c\right|, \left| \eta
\right|: c,\eta \in {\mathbb C}}\Delta_l(c,\eta,\mu)\leq g(x_{0},y_{0})+ \displaystyle%
\lim_{V_l\to\infty}\{-g^*_l(x^*_{\mathbf{0},l},y^*_{\mathbf{0},l})\}=
0.
\end{equation}

(ii) We have that:
\begin{eqnarray}
 \frac{1}{V_l}\langle  \hat{H}^{\mathrm{appr}}_l (c,\eta,\mu) - \hat{H}_l(\mu) \rangle_{\hat{H}_l(\mu)}
 = -(\lambda +\mu ) \left|c\right|^2 + \gamma_0\left|c\right|^{4}
 \nonumber\\+  (\lambda -\mu ) \left|\eta\right|^2 + \gamma_0\left|\eta\right|^{4}
 + \langle  (\lambda + \mu)\hat{\rho}_{\mathbf{0},-} - \gamma_0\hat{\rho}_{\mathbf{0},-}
 (\hat{\rho}_{0,-}
 - \frac{1}{V_l})\rangle_{\hat{H}_l(\mu)} \nonumber\\+ \langle  (\mu-\lambda)\hat{\rho}_{\mathbf{0},+} - \gamma_0\hat{\rho}_{\mathbf{0},+}
 (\hat{\rho}_{\mathbf{0},+}
 - \frac{1}{V_l})\rangle_{\hat{H}_l(\mu)}
 \end{eqnarray}
Let $ \rho_{\mathbf{0},\sigma} =
\langle\hat{\rho}_{\mathbf{0},\sigma} \rangle_{\hat{H}_l(\mu)}$.
From the above result and noting that
\begin{equation}
\displaystyle\inf_{|c|,|\eta|:c,\eta\in{\mathbb C}}\{p_l
(\beta,\mu) - p^{\mathrm{appr}}_l (\beta, c,\eta, \mu)\} =  p_l (\beta,\mu)-
\displaystyle\sup_{\left|c\right|,\left|\eta\right|: c,\eta\in C}
p^{\mathrm{appr}}_l (\beta, c, \mu)
\end{equation}
and $\langle
(\hat{\rho}_{\mathbf{0},\sigma})^2 \rangle_{\hat{H}_l(\mu)}\geq
(\langle \hat{\rho}_{\mathbf{0},\sigma}\rangle_{\hat{H}_l(\mu)})^2 $
it follows that,
\begin{eqnarray}
 0\leq   p_l (\beta,\mu)-  \displaystyle\sup_{\left|c\right|:c\in {\mathbb C}} p^{\mathrm{appr}}_l
 (\beta, c,\eta, \mu) \leq
 \displaystyle\inf_{\left|c\right|: c\in C}\{-(\lambda +\mu ) \left|c\right|^2 + \gamma_0\left|c\right|^{4}\}
 +&\nonumber\\+
 \displaystyle\inf_{\left|\eta\right|: \eta\in C}\{(\mu-\lambda ) \left|\eta\right|^2 + \gamma_0\left|\eta\right|^{4}\}+   (\mu+\lambda + \frac{\gamma_0}{V_l})
 \rho_{\mathbf{0},-}-\gamma_0\rho^2_{\mathbf{0},-} + &\nonumber\\+ (\mu-\lambda+ \frac{\gamma_0}{V_l})
 \rho_{\mathbf{0},+}-\gamma_0\rho^2_{\mathbf{0},+}.
\end{eqnarray}
Note that
\begin{equation}
 \displaystyle\sup_{\rho_{\mathbf{0},\sigma}\geq 0} \{
(\mu-\sigma\lambda +
\frac{\gamma_0}{V_l})\rho_{\mathbf{0},\sigma}-\gamma_0\rho^2_{\mathbf{0},\sigma}\}=
\left\{\begin{array}{ll}
 0 ,& \mbox{ $ \mu \leq \sigma\lambda - \dfrac{\gamma_0}{V_l} $}, \\
\noalign{\smallskip} \dfrac{\left(\mu-\sigma\lambda +
\dfrac{\gamma_0}{V_l}\right)^2}{4\gamma_0},&\mbox{ $\mu > \sigma\lambda -
\dfrac{\gamma_0}{V_l}.$}
\end{array}\right.
\end{equation}
On the other hand
\begin{equation}
 \displaystyle\inf_{x\geq 0} \{
(\sigma\lambda-\mu )x+\gamma_0 x^2\}= \left\{\begin{array}{ll}
 0 ,& \mbox{ $ \mu \leq \sigma\lambda  $}, \\
\noalign{\smallskip} \dfrac{(\mu-\sigma\lambda)^2}{4\gamma_0},&\mbox{
$\mu > \sigma\lambda.$}
\end{array}\right.
\end{equation}
The respective  maximum and minimum are attained in the following
way,
\begin{equation}\label{eq20}
\rho^*_{\mathbf{0},\sigma}= \left\{\begin{array}{ll}
 0 ,& \mbox{ $ \mu \leq \sigma\lambda -\dfrac{\gamma_0}{V_l}  $} \\
\noalign{\smallskip} \dfrac{\mu-\sigma\lambda}{2\gamma_0},&\mbox{
$\mu
> \sigma\lambda -\dfrac{\gamma_0}{V_l},$}
\end{array}\right.\quad
 x_{\mathbf{0},\sigma}= \left\{\begin{array}{ll}
 0 ,& \mbox{ $ \mu \leq \sigma\lambda  $} \\
\noalign{\smallskip} \dfrac{\mu-\sigma\lambda}{2\gamma_0},&\mbox{
$\mu
> \sigma\lambda.$}
\end{array}\right.
\end{equation}
Combining these results we obtain, in the thermodynamic limit, the
proof for the case $\gamma = 0.$

(iii) In the case $ 2\gamma_0 = \gamma $ we have $ g(x,y)= -(\lambda +
\mu)x + \gamma_0(x+y)^2 + (\lambda-\mu)y.$

(a) $\mu \in (-\infty, -\lambda].$ In this case $g(x,y)\geq 0,$ for
all $(x,y)\in ({\mathbb R}^+)^2,$ \\then
$\displaystyle\inf_{(x,y)\in ({\mathbb R}^+)^2}\{g(x,y)\} = g(0,0)=
0,$ where the minimum  is attained at \begin{equation}\label{eq21}
x_{\mathbf{0}}= 0, y_{\mathbf{0}}= 0.\end{equation}

(b) $\mu \in (-\lambda, \lambda].$  We have,
\begin{eqnarray}
&\displaystyle\inf_{(x,y)\in ({\mathbb R}^+)^2}\{g(x,y)\} =
\displaystyle\inf_{(x,y)\in ({\mathbb R}^+)^2}\{ -(\lambda + \mu)x +
\gamma_0(x+y)^2 + (\lambda-\mu)y  \}\nonumber\\&= \displaystyle\inf_{x\in
{\mathbb R}^+}\{g(x,0)\} =\displaystyle\inf_{x\in {\mathbb R}^+}\{
-(\lambda + \mu)x + \gamma_0 x^2\} = -\frac{(\mu + \lambda
)^2}{4\gamma_0},
\end{eqnarray}
where the minimizer is
given by,
\begin{equation}\label{eq22}
x_{\mathbf{0}}=\frac{\mu + \lambda}{2\gamma_0},\;y_{\mathbf{0}}=
0.
\end{equation}

(c) $\mu \in (\lambda, \infty).$ Let $r_1 = \lambda + \mu,$
$r_2 = \mu-\lambda.$ Then,
\begin{equation}
g(x,y) = -r_1x + \gamma_0(x+y)^2 -r_2y.
\end{equation}
In this case,
\begin{equation}
\displaystyle\inf_{(x,y)\in
({\mathbb R}^+)^2}\{g(x,y)\} = -\frac{\lambda \mu}{\gamma_0}.
\end{equation}
The minimizer is given by,
\begin{equation}\label{eq23}
x_{\mathbf{0}} = \frac{\mu }{\gamma_0},\;y_{\mathbf{0}} = 0.
\end{equation}
This ends the proof.
\end{proof}

\begin{corollary}\label{thm5} Assume that $\sup_{r\geq 0}\left\{
(\mu-\alpha) r-g(r) \right\}$ exists. Then,
\begin{equation}
p (\beta,\mu)= \left\{\begin{array}{ll}
\dfrac{\lambda^2}{2\gamma_0-\gamma}
+\dfrac{\mu^2}{2\gamma_0+\gamma} +p^{\mathrm{MF}^{'}}(\beta,\mu),&
\mbox{$ 2\gamma_0 > \gamma > 0,\mu \geq \dfrac{ \lambda (2\gamma_0
+\gamma)}{2\gamma_0 - \gamma},$}\\ \noalign{\smallskip} h_1
(\lambda,\mu) +
p^{\mathrm{MF}^{'}}(\beta,\mu),&\mbox{ $\gamma = 0,\;\gamma_0\neq 0$}, \\
\noalign{\smallskip} h_2 (\lambda,\mu) +
p^{\mathrm{MF}^{'}}(\beta,\mu),&\mbox{ $2\gamma_0 = \gamma$},
\end{array}\right.
\end{equation}
where,
\begin{equation}
h_1 (\lambda,\mu)= \left\{\begin{array}{ll}
 0 ,&\mbox{$\mu \in (-\infty, -\lambda]$} \\
\noalign{\smallskip} \dfrac{(\mu + \lambda )^2}{4\gamma_0},&\mbox{
$\mu \in (-\lambda, \lambda]$}\\
\noalign{\smallskip}  \dfrac{\mu^2 + \lambda^2}{2\gamma_0},&\mbox{
$\mu \in (\lambda, \infty)$}
\end{array}\right.
 h_2 (\lambda,\mu)= \left\{\begin{array}{ll}
 0 ,& \mbox{$\mu \in (-\infty, -\lambda]$} \\
\noalign{\smallskip} \dfrac{(\mu + \lambda )^2}{4\gamma_0},&\mbox{
$\mu \in (-\lambda, \lambda]$}\\
\noalign{\smallskip} \dfrac{\lambda\mu}{\gamma_0},&\mbox{ $\mu \in
(\lambda, \infty)$}.
\end{array}\right.
\end{equation}
and
\begin{equation}
p^{\mathrm{MF}^{'}}(\beta,\mu) = \displaystyle\sup_{r\geq 0}\inf_{\alpha\leq
0}\left\{ (\mu-\alpha )r-g(r) + p^{id^{'}}(\beta,\alpha)\right\},
\end{equation}
being $ p^{\mathrm{id}^{'}}(\beta,\alpha)$ the limit grand canonical pressure for
the ideal Bose gas with excluded zero mode at chemical potential
$\alpha \leq 0.$
\end{corollary}

\begin{proof}
Let $ \hat{H}^{'}(\mu) =\hat{H}^{\mathrm{MF}^{'}}(\mu)
+\frac{\gamma_0}{V_l}\sum_{\sigma,\mathbf{p}\in\Lambda^*_l\backslash
\{0\}}
(\hat{a}^{\dag}_{\mathbf{p},\sigma})^2\hat{a}^2_{\mathbf{p},\sigma}
= \hat{H}^{\mathrm{MF}^{'}}(\mu) + \frac{\gamma_0}{V_l}\hat{D}$ and
$ \hat{H}^{0^{'}}_l =
\sum_{\sigma,p\in\Lambda^*_l\backslash\{\mathbf{0}\}}
\lambda_{l,\sigma}(\mathbf{p})\hat{a}^{\dag}_{\mathbf{p},\sigma}\hat{a}_{\mathbf{p},\sigma}.$
Let $f^{\Gamma}_l (\varrho) $ be the canonical energy at finite
volume $V_l$ and density $\varrho$ associated with the Hamiltonian
$\Gamma,$ defined as,
\begin{equation} f^{\Gamma}_l (\varrho)
= -\frac{1}{V_l}\ln \Tr_{{\mathcal H}^{(\mathrm B,N)}_l} e^{-\beta
\Gamma|_{{\mathcal H}^{(\mathrm B,N)}_l}},
\end{equation} where ${\mathcal H}^{(\mathrm B,N)}_l $ is the Hilbert space representing a $N$-particles
Bose system. Let $\varrho_0$ the density of particles in the zero
mode.
Let $f^{'}_l(\varrho-\varrho_0), f^{\mathrm{MF}^{'}}_l
(\varrho-\varrho_{\mathbf{0}}),f^{\mathrm{id}^{'}}_l
(\varrho-\varrho_{\mathbf{0}}) $ be the finite canonical energies
associated to the operators $\hat{H}^{'}, \hat{H}^{\mathrm{MF}^{'}},
\hat{H}^{0'}_{l}$ respectively and
$f^{'}(\varrho-\varrho_{\mathbf{0}}),
f^{\mathrm{MF}^{'}}(\varrho-\varrho_{\mathbf{0}}),f^{\mathrm{id}^{'}}(\varrho-\varrho_{\mathbf{0}})
$ their respective limit canonical energies.

We only sketch the proof which can be found in Refs.11,12 in the
case $g(x) = ax^2.$ In this case the following identity can be
easily verified:

\begin{equation} f^{'}_l (\varrho-\varrho_{\mathbf{0}}) = g(\varrho-\varrho_{\mathbf{0}}) -
\frac{1}{\beta V_l}\ln\langle
e^{-\frac{\beta\gamma_0}{V_l}\hat{D}}\rangle_{\hat{H}^{0^{'}}_l(\varrho-\varrho_{\mathbf{0}})}
+ f^{\mathrm{id}^{'}}_l (\varrho-\varrho_{\mathbf{0}}),
\end{equation}
where
$\langle
\cdot\rangle_{\hat{H}_{0,l}(\varrho-\varrho_{\mathbf{0}})},$
represents the canonical Gibbs state at density
$\varrho-\varrho_{\mathbf{0}}$ associated to the ideal gas with the
single mode zero excluded.
Using the following facts,
\begin{equation}\langle e^{-\frac{\beta\gamma_0}{V_l}\hat{D}}\rangle_{\hat{H}^{0^{'}}_l(\varrho-\varrho_{\mathbf{0}})}
\geq e^{-\langle
\frac{\beta\gamma_0}{V_l}\hat{D}\rangle_{\hat{H}^{0^{'}}_l(\varrho-\varrho_{\mathbf{0}})}},
\end{equation}
\begin{equation}
e^{-\frac{\beta\gamma_0}{V_l}\hat{D}}\leq 1,
\end{equation}
and noting that
\begin{equation} \displaystyle\lim_{V_l\to\infty}\left\langle
\frac{\hat{D}}{V^2_l}\right\rangle_{\hat{H}^{0^{'}}_l(\varrho-\varrho_{\mathbf{0}})}
= 0,
\end{equation}
we get
\begin{equation}
\displaystyle\lim_{V_l\to\infty}\frac{1}{\beta V_l} \ln\langle
e^{-\frac{\beta\gamma_0}{V_l}\hat{D}}\rangle_{\hat{H}^{0^{'}}_l(\varrho-\varrho_{\mathbf{0}})}
= 0.
\end{equation}
This implies that in the thermodynamic limit
\begin{equation}
f^{'}(\varrho-\varrho_{\mathbf{0}}) = f^{\mathrm{MF}^{'}}(\varrho-\varrho_{\mathbf{0}}) =
g(\varrho-\varrho_{\mathbf{0}}) + f^{\mathrm{id}^{'}}
(\varrho-\varrho_{\mathbf{0}}).
\end{equation}
Then,
\begin{eqnarray}
& p^{'}(\beta,\mu) =
\displaystyle\sup_{\varrho\geq\varrho_{\mathbf{0}}}\{
\mu(\varrho-\varrho_{\mathbf{0}}) - f^{\mathrm{MF}^{'}}\}\nonumber\\& =
\displaystyle\sup_{\varrho\geq\varrho_{\mathbf{0}}}\{
\mu(\varrho-\varrho_{\mathbf{0}}) - g (\varrho-\varrho_{\mathbf{0}})
- f^{\mathrm{id}^{'}} (\varrho-\varrho_{\mathbf{0}})\},
\end{eqnarray}
but
\begin{equation}
f^{\mathrm{id}^{'}}(\varrho-\varrho_{\mathbf{0}})=\displaystyle\sup_{\alpha\leq 0}
\{ \alpha (\varrho-\varrho_{\mathbf{0}}) - p^{\mathrm{id}^{'}}(\alpha)\}.
\end{equation}
Therefore,
\begin{eqnarray}
& p^{'}(\beta,\mu) =
\displaystyle\sup_{\varrho\geq\varrho_{\mathbf{0}}}\displaystyle\sup_{\alpha\leq
0}\{ (\mu - \alpha)(\varrho-\varrho_{\mathbf{0}}) - g
(\varrho-\varrho_{\mathbf{0}}) +p^{\mathrm{id}^{'}} (\alpha)\}\nonumber\\&=
\displaystyle\sup_{r\geq 0}\displaystyle\sup_{\alpha\leq 0}\{ (\mu -
\alpha)(r) - g (r) +p^{\mathrm{id}^{'}} (\alpha)\}.
\end{eqnarray}
On the other hand it is clear that in the case $2\gamma_0 > \gamma$
we obtain,
\begin{eqnarray}
p(\beta,\mu) &=\displaystyle\sup_{(x,y)\in ({\mathbb
R}^+)^2}\left\{(\mu+\lambda)x -\gamma_0 x^2+ (\mu-\lambda)y
-\gamma_0 y^2 -\gamma xy\right\}+
p^{'}(\beta,\mu)\nonumber\\&=(\mu+\lambda)x_0 -\gamma_0 x^2_0+
(\mu-\lambda)y_0 -\gamma_0 y^2_0-\gamma x_0 y_0 + p^{'}(\beta,\mu)
\nonumber\\&= \left(\dfrac{\lambda^2}{2\gamma_0-\gamma}
+\dfrac{\mu^2}{2\gamma_0+\gamma}\right) +p^{'}(\beta,\mu).
\end{eqnarray}
In the case $ 2\gamma_0 = \gamma $ the proof follows from the
theorem~\ref{thm4}.

\end{proof}

\subsection{BEC}

\subsubsection{ Coexistence of two non-conventional Bose--Einstein
condensates}

In this section we theoretically predict the exciting possibility of
coe\-xisting non-conventional Bose--Einstein condensates.

\begin{theorem}
\label{thm6} Under the conditions $2\gamma_0 >\gamma $ and $\mu \geq
\lambda (2\gamma_0 +\gamma)/(2\gamma_0 - \gamma),$ the system given by ~(\ref%
{eq1}) displays simultaneous non-conventional  Bose--Einstein
condensation of the two levels associated to the zero mode and the
amounts of condensate are  given as:
\begin{equation}  \label{eq24}
\rho_{\mathbf{0},-}(\mu) = \frac{\lambda}{ 2\gamma_0-\gamma} +
\frac{\mu}{2\gamma_0 + \gamma}
\end{equation}
\begin{equation} \label{eq25}
\rho_{\mathbf{0},+}(\mu) = -\frac{\lambda}{ 2\gamma_0-\gamma} +
\frac{\mu}{2\gamma_0 + \gamma}.
\end{equation}
\end{theorem}
\begin{proof}
$p(\beta,\mu )$ can be rewritten as,
\begin{equation}
p(\beta,\mu) =
\frac{1}{4}\left(\frac{(\alpha_{-}-\alpha_{+})^2}{2\gamma_0-\gamma}
+\frac{(\alpha_{-}+\alpha_{+})^2}{2\gamma_0+\gamma}\right)+p^{\mathrm{MF}^{'}}(\beta,\mu),
\end{equation}
where $ \alpha_{-} = \lambda + \mu,$ $\alpha_{+} = \mu -\lambda.$
Using the convexity of $p_l(\beta,\mu)$ and
$p^{\mathrm{appr}}_l(\beta,c,\eta,\mu)$ respect to $\alpha_{-}$ and
$\alpha_{+}$ we get from the Griffiths theorem~\cite{griff}
and theorem~\ref{thm4},
\begin{eqnarray}
\rho_{\mathbf{0},\sigma}&=
\displaystyle\lim_{V_l\to\infty}\partial_{\alpha_{\sigma}}
p_l(\beta,\mu) =\partial_{\alpha_{\sigma}}\displaystyle\lim_{V_l\to\infty}
p_l(\beta,\mu)=\partial_{\alpha_{\sigma}}\displaystyle\lim_{V_l\to\infty}
\displaystyle\sup_{\left|c\right|,\left|\eta\right|: c,\eta\in
{\mathbb C}}p_l(\beta,c,\eta,\mu)\nonumber\\&=\partial_{\alpha_{\sigma}}
\left(
\dfrac{1}{4}\left(\dfrac{(\alpha_{-}-\alpha_{+})^2}{2\gamma_0-\gamma}
+\dfrac{(\alpha_{-}+\alpha_{+})^2}{2\gamma_0+\gamma}\right)+p^{\mathrm{MF}^{'}}(\beta,\mu)
\right),
\end{eqnarray}
for $\sigma =-,+.$ An analogous procedure leads to eq.~(\ref{eq25}) This yields to the proof
of the theorem.
\end{proof}

\begin{corollary}\label{thm7}
For $\gamma = 0$  we have,
\begin{equation}
\rho_{\mathbf{0},-}= \left\{\begin{array}{ll}
0,&
\mbox{ $\mu \in (-\infty, -\lambda ]$} \\
\noalign{\smallskip} \dfrac{\mu +\lambda}{2\gamma_0},&\mbox{ $\mu \in
(-\lambda, \infty),$}
\end{array}\right.\quad
\rho_{\mathbf{0},+}= \left\{\begin{array}{ll} 0,&
\mbox{ $\mu \in (-\infty, \lambda ]$} \\
\noalign{\smallskip} \dfrac{\mu -\lambda}{2\gamma_0},&\mbox{ $\mu \in
(\lambda, \infty)$}
\end{array}\right.\end{equation} and for $\gamma = 2\gamma_0,$
\begin{equation}\rho_{\mathbf{0},-}= \left\{\begin{array}{ll}
0,&
\mbox{ $\mu \in (-\infty, -\lambda ]$} \\
\noalign{\smallskip} \dfrac{\mu +\lambda}{2\gamma_0},&\mbox{ $\mu \in
(-\lambda, \lambda],$} \\
\noalign{\smallskip} \dfrac{\mu}{\gamma_0},&\mbox{ $\mu \in (\lambda,
\infty],$}
\end{array}\right.\quad
\rho_{\mathbf{0},+}= 0, \;\;\;\mu \in {\mathbb R}.
\end{equation}
\end{corollary}

\begin{proof} This proof is a direct consequence of
Corollary ~\ref{thm5}.

\end{proof}

Let $\tau_{\mathrm{ss},l} (\beta,\mu)$ and $\tau_{\mathrm{cs},l}$ be  thermal averages
related to the self-scat\-te\-ring and cross-scattering operators
defined as:
\begin{equation}
\tau_{\mathrm{ss},l} (\beta,\mu) = V^{-1}_l
\displaystyle\sum_{\sigma}\langle(\hat{a}^{\dag}_{\mathbf{0},\sigma})^2\hat{a}^2_{\mathbf{0},\sigma}
 \rangle_{\hat{H}_l(\mu)},\;\;
\tau_{\mathrm{cs},l} (\beta,\mu) = \langle
\hat{\rho}_{\mathbf{0},+}\hat{\rho}_{\mathbf{0},-}
\rangle_{\hat{H}_l(\mu)}.
\end{equation}

\begin{corollary}\label{thm8}
For $2\gamma_0 > \gamma $ and $\mu \geq \lambda (2\gamma_0
+\gamma)/(2\gamma_0 - \gamma), $

\begin{eqnarray}
\displaystyle\lim_{V_l\to\infty}V^{-1}_l\tau_{\mathrm{cs},l} (\beta,\mu) & =
\displaystyle\lim_{V_l\to\infty}\langle
\hat{\rho}_{0,+}\hat{\rho}_{0,-} \rangle_{\hat{H}_l(\mu)}=
\displaystyle\lim_{V_l\to\infty}\langle
\hat{\rho}_{0,+}\rangle_{\hat{H}_l(\mu)}\langle\hat{\rho}_{0,-}
\rangle_{\hat{H}_l(\mu)}\nonumber \\& =
\left(\dfrac{\mu}{2\gamma_0+\gamma}\right)^2
-\left(\dfrac{\lambda}{2\gamma_0-\gamma}\right)^2.
\end{eqnarray}

\end{corollary}

\begin{corollary}\label{thm9} For
$2\gamma_0
> \gamma > 0,\;\mu \geq \lambda (2\gamma_0
+\gamma)/(2\gamma_0 - \gamma),$

\begin{equation}
\displaystyle\lim_{V_l\to\infty}V^{-1}_l\tau_{\mathrm{ss},l} (\beta,\mu)=
 2\left(\left(\dfrac{\lambda}{2\gamma_0-\gamma}\right)^2
+\left(\dfrac{\mu}{2\gamma_0+\gamma}\right)^2\right),
\end{equation}
and for $ \gamma = 0,$
\begin{equation}
\displaystyle\lim_{V_l\to\infty}V^{-1}_l\tau_{\mathrm{ss},l} (\beta,\mu)
=\left\{\begin{array}{ll}
 0 ,&\mbox{$\mu \in (-\infty, -\lambda]$} \\
\noalign{\smallskip} \left(\dfrac{\mu +
\lambda}{2\gamma_0}\right)^2,&\mbox{
$\mu \in (-\lambda, \lambda]$}\\
\noalign{\smallskip}  \dfrac{\mu^2 + \lambda^2}{2\gamma^2_0},&\mbox{
$\mu \in (\lambda, \infty)$}.
\end{array}\right.
\end{equation}
\end{corollary}

\begin{corollary}\label{thm10} For $\gamma = 2\gamma_0$ the following
identities

\begin{equation}
\displaystyle\lim_{V_l\to\infty}V^{-1}_l(\tau_{\mathrm{ss},l}(\beta,\mu)+
2\tau_{\mathrm{cs},l}
(\beta,\mu))=\left\{\begin{array}{ll}
 0 ,&\mbox{$\mu \in (-\infty, -\lambda]$} \\
\noalign{\smallskip} \left(\dfrac{\mu +
\lambda}{2\gamma_0}\right)^2,&\mbox{
$\mu \in (-\lambda, \lambda]$}\\
\noalign{\smallskip}  \dfrac{\lambda\mu}{\gamma^2_0},&\mbox{ $\mu \in
(\lambda, \infty)$}
\end{array}\right.
\end{equation}
hold.
\end{corollary}

\section{Generalized BEC}\label{sec4}

We consider the case $g(x) = ax^2, a >0.$
In this case Hamiltonians ~(\ref{eq1}) and $\hat{H}^{\mathrm{MF}^{'}}_l(\mu
)$ become

\begin{equation}\hat{H}_l (\mu) =
\displaystyle\sum_{\sigma,\mathbf{p}\in\Lambda^*_l\backslash
\{0\}}\left((\lambda_{l,\sigma}(\mathbf{p})
-\mu_l)\hat{n}_{\mathbf{p},\sigma} +
\frac{\gamma_0}{V_l}\hat{n}^2_{\mathbf{p},\sigma}\right) +
\frac{a}{V_l}(\hat{N}^{'})^2,
\end{equation}
\begin{equation}\hat{H}^{\mathrm{MF}^{'}}_l (\mu) =
\displaystyle\sum_{\sigma,\mathbf{p}\in\Lambda^*_l\backslash
\{0\}}(\lambda_{l,\sigma}(\mathbf{p})
-\mu_l)\hat{n}_{\mathbf{p},\sigma} + \frac{a}{V_l}(\hat{N}^{'})^2,
\end{equation}
respectively.

Let $p^{'}_l(\beta,\mu),$ $ p^{\mathrm{MF}^{'}}_l (\beta,\mu )$ be
the grand canonical finite pressures associated to these operators.
For every $\mu \in {\mathbb R},$ in the thermodynamic limit, it is
proved in Ref.17  that
\begin{eqnarray}
&\displaystyle\lim_{V_l\to\infty}p^{'}_l(\beta,\mu) =
p^{'}(\beta,\mu)= \displaystyle\lim_{V_l\to\infty}p^{\mathrm{MF}^{'}}_l
(\beta,\mu ) \nonumber\\& = p^{\mathrm{MF}^{'}}(\beta,\mu)=\displaystyle\sup_{\alpha
\leq 0}\left\{ p^{\mathrm{id}}(\beta,\alpha) + \frac{(\mu-\alpha)^2}{2a}
\right\}.
\end{eqnarray}

\begin{proposition}\label{thm11} Let
\begin{equation}
\rho^{\mathrm{id}}_{\mathrm{c}} (\beta) = \dfrac{2}{(2\pi)^d}
\int_{\left[0,\infty\right)^d}\dfrac{d^d\mathbf{p}}{ e^{\beta
\dfrac{\left\|\mathbf{p}\right\|^2}{2}} -1}
\end{equation}
and
$2\gamma_0
>\gamma, $ then

\begin{equation}
p(\beta,\mu)= \left\{\begin{array}{ll}
 \dfrac{\lambda^2}{2\gamma_0-\gamma}
+\dfrac{\mu^2}{2\gamma_0+\gamma} + p^{\mathrm{MF}^{'}}(\beta,\mu) ,&
\mbox{ $ \dfrac{
\lambda (2\gamma_0 +\gamma)}{2\gamma_0 - \gamma} \leq \mu \leq \dfrac{a}{2}\rho^{\mathrm{id}}_{\mathrm{c}} (\beta) $} \\
\noalign{\smallskip} \dfrac{\lambda^2}{2\gamma_0-\gamma}
+\dfrac{\mu^2}{a}+\dfrac{\mu^2}{2\gamma_0+\gamma} +
p^{\mathrm{id}}(\beta,0) ,&\mbox{ $\mu \geq \dfrac{ \lambda (2\gamma_0
+\gamma)}{2\gamma_0 - \gamma} > \dfrac{a}{2}\rho^{\mathrm{id}}_{\mathrm{c}}(\beta).$}
\end{array}\right. \end{equation}

\end{proposition}

Moreover this model displays
absence of macroscopic occupation of nonzero single levels:
\begin{equation}\displaystyle\lim_{V_l\to\infty}\left\langle
\frac{\hat{n}_{\mathbf{p},\sigma}}{V_l}\right\rangle_{\hat{H}_l(\mu)}
=0\end{equation} for $\sigma = +,- $ and
$\mathbf{p}\in\Lambda^*_l\backslash\{\mathbf{0}\},$ but generalized
BEC can be verified in the sense that
\begin{equation}\displaystyle \lim_{\delta\to 0^+}\displaystyle
\lim_{V_l\to\infty}\sum_{\sigma,\mathbf{p}\in\Lambda^*_l:\lambda_l(\mathbf{p})
\leq \delta}\left\langle
\frac{\hat{n}_{\mathbf{p},\sigma}}{V_l}\right\rangle_{\hat{H}%
_l(\mu)} > 0,\end{equation}
being $\hat{n}_{p} =
\hat{n}_{\mathbf{p},-} +\hat{n}_{\mathbf{p},+}.$

\section{Conclusions}

The pressure of the model given by ~(\ref{eq1}) can be exactly
determined for a range of values of the chemical potential $\mu$ by
using the energy operator ~(\ref{eq10}) in the framework of the
Bogolyubov approximation. The superstable system displays
non-conventional BEC consisting in the macroscopic occupation of
both ground state \-levels.The condensates can eventually coexist
depending on the values of $\mu.$ Moreover, as a subtle consequence
of these results we have determined the exact values of some thermal
averages associated to the cross scattering and self-scattering
operators. Finally, in the case $g(x)= ax^2,$ $a
>0,$ generalized BEC holds.

\section*{Acknowledgements}

This work has been partially
supported by Grant PBCT-ACT13 (Stochastic Analysis Laboratory,
Chile), Programa de Mag\'{\i}ster en Matem\'aticas, Universidad
de La Serena and Program ``Fundamental problems of nonlinear dynamics''  of the RAS.

\end{document}